\documentclass[a4paper,10pt]{article}

\usepackage[english]{babel}
\usepackage[ansinew]{inputenc}
\usepackage{amsmath}
\usepackage{amssymb}
\usepackage{amsthm}
\usepackage{stmaryrd}
\usepackage{mathrsfs}
\usepackage{enumerate}
\usepackage{dsfont}
\usepackage{amsthm}
\usepackage{geometry}
\usepackage{graphicx}
\usepackage{wrapfig}
\usepackage{accents}
\usepackage{array}
\usepackage{cite}
\usepackage{color}
\usepackage{hyperref}
\usepackage{esint}
\usepackage{array}
\usepackage{tikz}
\usepackage{float}
\usetikzlibrary{arrows}

\usepackage{fancyhdr}
\makeatletter
\renewcommand\@biblabel[1]{#1.}
\makeatother

\numberwithin{equation}{section}

\newcommand{\mH}{\mathcal{H}}

\newcommand{\tI}{\tilde{I}}

\newcommand{\R}{\mathds{R}}
\newcommand{\Z}{\mathds{Z}}

\newcommand{\hG}{\hat{\Gamma}}

\newcommand{\ph}{\varphi}
\renewcommand{\o}{\omega}
\newcommand{\e}{\varepsilon}
\renewcommand{\l}{\lambda}
\renewcommand{\a}{\alpha}
\newcommand{\beq}{\begin{equation*}}
\newcommand{\eeq}{\end{equation*}}
\newcommand{\bneq}{\begin{equation}}
\newcommand{\eneq}{\end{equation}}

\newtheorem{lemma}{Lemma}[section]
\newtheorem{theorem}{Proposition}[section]

\author{Matthias Plaschke \\ \\ \small{Faculty of Physics, University of Vienna}}
\title{Local Anyonic Quantum Fields on the Circle leading to Cone-Local Anyons in Two Dimensions}
\date{\today}

\begin{document}
\maketitle

\begin{abstract}
    Using the method of implementable one-particle Bogoliubov transformations it is possible to explicitly define a local covariant net of quantum fields on the (universal covering of the) circle $S_1$ with braid group statistics. These Anyon fields transform under a representation of $\widetilde{U(1)}$ for arbitrary real-valued spin and their commutation relations depend on the relative winding number of localization regions. By taking the tensor product with a local covariant field theory on $\R^2$ one can obtain a (non-boost-covariant) cone-localized field net for Anyons in two dimensions.
\end{abstract}

%

\vspace{2em}
\section{Introduction}
In $d=1+1$ dimensions, localized \emph{massive} quantum fields with Anyonic commutation relations can be explicitly constructed by defining ``order'' and ``disorder'' operators, using implementable Bogoliubov transformations on the 1-particle space (e.g. \cite{Adler96,Salv07,Ruij81} etc.).
In $2+1$ dimensions, models exhibiting Anyon-like features often start with a classical gauge field theory with additional Chern-Simons and Higgs terms, which then leads to fields carrying electric charge and magnetic flux. It is believed that such models can be used in the description of high-temperature superconductivity, the quantum Hall effect or other two-dimensional physical systems. However, such constructions can only be made rigorous by using, for example, a lattice approximation (see e.g. \cite{FroeMar89}).
On the other hand more explicit examples of Anyon quantum fields as for instance in \cite{Sem88, LigMinRos93, Marino93,GrazRoth94} lead to formal commutation relations of the form
\beq
\phi(x_1) \phi(x_2) \sim e^{i\pi\lambda \operatorname{sign}[\operatorname{arg}(x_1-x_2)]} \phi(x_2) \phi(x_1),
\eeq
where $\operatorname{arg}(x)$ denotes the angle of the vector $x \in \R^2$. After smearing the point-fields with testfunctions such commutation relations would lead to localization of the operators in smeared out ``double strings'', extending to infinity in two opposite directions. However, to speak of ``proper'' Anyons in the sense of algebraic quantum field theory one would need field algebras localized in spacelike cones $C$, extending to infinity only in a connected compact set of spacelike directions. More precisely, the localization regions for Anyons can be labelled by paths $\tilde{C}$ of such spacelike cones depending also on a kind of winding number which determines the commutation relations of mutually spacelike separated fields \cite{BuchFred82, FreRehSchr92, FroeMar91, FreGabRue96, Mund03}(instead of heaving a mere sign-function in the exchange phase). The spin statistics theorem for Anyons \cite{FroeMar91,Mund09} then forces these fields to transform under a representation of the (covering of the) Poincaré group $\widetilde{\mathcal{P}}_+^\uparrow$ with arbitrary real spin $s\in\R$. \\
It would be desirable to find an explicit construction for such cone-localized Anyon fields for any spin $s$ in $2+1$ dimensions. Unfortunately a No-Go theorem by Bros and Mund \cite{BrosMund12} states that there are no \emph{free} fields for Anyons. More precisely, they prove that the scattering matrix for a relativistic cone-local model for Anyons (satisfying in addition the Reeh-Schlieder property) always has to be non-trivial. This also means that there are no Anyon fields which simply create single particle states from the vacuum \cite{Mund98}. This makes the explicit non-perturbative construction of Anyons very complicated. \\
It is, however, possible to circumvent this No-Go theorem by considering fields which are only localized in wedge regions in Minkowski space. In this case one can modify the method of multiplicative deformations, established in \cite{LechGro08,Lech12}, to obtain one-particle generators localized in so-called ``paths of wedges'' which have anyonic commutation relations and are covariant w.r.t a representation of $\widetilde{\mathcal{P}}_+$ for spin $s\in\R$ \cite{Pla13}. This is only possible because in the proof of the No-Go theorem one needs to have three mutually spacelike separated localization regions, which is impossible for wedges. \\
The No-Go theorem now makes it very difficult to explicitly construct Anyon fields in $d=2+1$ with sharper localization. To better understand the concepts of winding number, anyonic commutation relations and arbitrary spin and their connection, we will therefore consider the simplified case of Anyon quantum fields on the circle. A similar construction using also the current algebra on the circle (see also \cite{BuchMackTod88}) can be found in \cite{CarLang98,CarLang10} where, however, the phase factor appearing in the commutation relations is again governed by a sign-function. After smearing the point-fields this leads to localization in two disconnected regions lying on opposite sides of the circle. Here we will modify this construction in such a way that one obtains a compactly localized field algebra, covariant under a real spin representation of the rotations and with commutation relations depending on the relative winding number of the respective localization regions.\\
By taking tensor products of these fields on the circle with local fields on $\R^{2+1}$ one arrives at a rotation- and translation covariant cone-localized quantum field theory with anyonic commutation relations. However, due to the simple tensor product structure, covariance w.r.t. boosts is lost in the construction.

\vspace{1em}
\section{Aim} \label{aim}
Because the spin for Anyons is an arbitrary real number, rotations around multiples of $2\pi$ will act non-trivially on our fields. We will therefore consider the fields to be localized in intervals $\tI$ on the \emph{universal covering} $\widetilde{S_1}$ of the circle with radius $R=1$ (the generalization to any $R>0$ being straightforward). Our aim is to construct a net $\tI \mapsto \mathcal{F}(\tI)$ of $*$-algebras of operators acting on a Hilbert space $\mH$ satisfying the following properties, which define an anyonic field net on the circle\footnote{A related construction also regarding nets of von Neumann algebras on covering spaces of the circle has recently been provided in the master thesis of F. Bonesi. This idea of construction, using modular theory of standard subspaces of Hilbert spaces, is essentially contained in the lecture notes on conformal nets of Roberto Longo (see \url{http://www.mat.uniroma2.it/~longo/Lecture_Notes.html}).}:
\begin{enumerate}
    \item {\bf Charge Sectors:} The Hilbert space splits into a direct sum of Hilbert spaces for fixed charge,
        \bneq
        \mH = \bigoplus_{q\in\Z} \mH_q,
        \eneq
        and $\mH_0$ contains a unique rotation invariant vacuum vector $\Omega$. Moreover, the local algebras $\mathcal{F}(\tI)$ are generated by basic fields $\Phi$ that change the charge of a vector by $1$, i.e. $\Phi \mH_q \subset \mH_{q+1}$ and $\Phi^* \mH_q \subset \mH_{q-1}$.
    \item {\bf Isotony:} The map $\tI \mapsto \mathcal{F}(\tI)$ preserves inclusions, i.e. $\mathcal{F}(\tI_1) \subset \mathcal{F}(\tI_2)$ if $\tI_1 \subset \tI_2$.
    \item {\bf Covariance and Spin:} The field algebras are covariant under a representation $U$ of the universal covering group of the rotations $\widetilde{U(1)}\simeq \R$, i.e.
        \bneq
        U(\o)\mathcal{F}(\tI) U(\o)^* \subset \mathcal{F}(\tI+\o), \hspace{1em} \o \in \R
        \eneq
        and $2\pi$ rotations act as
        \bneq \label{2PiRep}
        U(2\pi) = \sum_{q\in\Z} e^{2\pi i S_q} P_q ,
        \eneq
        where $P_q$ is the projector onto the charge $q$ subspace and $S_q \in \R$ is the spin of the sector with charge $q$ (defined only modulo $1$).
    \item {\bf (Twisted) Locality:} Basic fields $\Phi_1, \Phi_2$ localized in intervals $\tI_1, \tI_2$, whose projections $I_1, I_2$ onto the base space $S_1$ do not intersect, satisfy commutation relations of the form
        \bneq \label{CommRel}
        \Phi_1\ \Phi_2 = \pm e^{2\pi i s (2 N(\tI_1,\tI_2)+1)}\ \Phi_2\ \Phi_1,
        \eneq
        where $s\in\R$ is a real parameter and $N(\tI_1,\tI_2) \in \Z$ is the relative winding number of $\tI_1$ w.r.t. $\tI_2$, which for consistency reasons has to satisfy $N(\tI_2,\tI_1) = - N(\tI_1,\tI_2)-1$ and $N(\tI_1+2\pi,\tI_2) = N(\tI_1,\tI_2)+1$.
\end{enumerate}
\vspace{1em}
\emph{Remarks:} \\
\emph{i)} We allow in general for a $\pm$ sign in equation \eqref{CommRel} because as we will see in the subsequent explicit construction we will arrive at commutation relations with an additional minus sign in front of the exchange phase factor. \\
\emph{ii)} In contrast to free fields for Bosons or Fermions the vector $\Phi \Omega$ will \emph{not} be a single particle state, but rather contain vectors with arbitrary high particle number. \\
\emph{iii)} Property iv) shows that the commutation relations will not be governed by a simple two-valued sign-function of the localization points, but depend on the winding number of two disjoint regions which makes sense because we work on the universal covering space $\widetilde{S_1}$ of the circle. \\

\section{Representation of the Rotations}
Transforming nontrivially under multiples of $2\pi$-rotations is one of the most important features of Anyons. This is why we have chosen to construct Anyon fields on the circle, because this one-dimensional model will be simple enough to be explicitly constructed by well-known techniques but is still capable of showing this non-trivial rotational behavior. \\
General considerations in algebraic quantum field theory show that the spins $S_q$ also have to satisfy the condition $S_{-q} = S_q$, which simply means that particles have the same spin as their corresponding anti-particles \cite{FroeMar91,Mund09}. Moreover, using the fact that the vacuum sector $\mH_0$ should be invariant under $2\pi$-rotaions, i.e. $S_0=0$, one can prove the following lemma.
\begin{theorem}
Assume we have a net of algebras according to section \ref{aim} and that the spins $S_q$ of the charged sectors $\mH_q$ for the representation $U$ of $\widetilde{U(1)}$ \eqref{2PiRep} additionally satisfy $S_{-q}=S_q$ and $S_0=0$. Then the $S_q$ are quadratic in $q$, i.e.
\bneq
S_q = s q^2,
\eneq
where $s\in\R$ is simply called the spin of the model.
\end{theorem}
\begin{proof}
Consider fields $\Phi_1, \Phi_2$ localized in $\tI_1, \tI_2$ with $I_1\cap I_2 = \emptyset$. Rotating the first field around $2\pi$ leads to the commutation relation
\beq
\begin{split}
U(2\pi)\Phi_1 U(2\pi)^*\ \Phi_2 &= \pm e^{2\pi i s(2N(\tI_1+2\pi,\tI_2)+1)}\Phi_2 \ U(2\pi)\Phi_1 U(2\pi)^* \\
&= \pm e^{2\pi i s (2N(\tI_1,\tI_2)+1)}\Phi_2 \ U(2\pi)\Phi_1 U(2\pi)^*\ e^{4\pi i s}.
\end{split}
\eeq
Using the shorthand notation $S_Q := \sum S_q P_q$ and the fact that $2\pi$-rotations are represented according to $U(2\pi) = e^{2\pi i S_Q}$ we can write
\beq
U(2\pi)\Phi_1 U(2\pi)^* = \Phi_1 e^{2\pi i (S_{Q+1}-S_Q)},
\eeq
which leads to the commutation relation
\beq
\begin{split}
U(2\pi)\Phi_1 U(2\pi)^*\ \Phi_2 &= \Phi_1 e^{2\pi i (S_{Q+1}-S_Q)}\ \Phi_2 \\
&= \pm e^{2\pi i s(2N(\tI_1,\tI_2)+1)} \Phi_2\ \Phi_1 e^{2\pi i (S_{Q+1}-S_Q)} \
e^{2\pi i (S_{Q+2}-2S_{Q+1}+S_Q)} \\
&= \pm e^{2\pi i s(2N(\tI_1,\tI_2)+1)} \Phi_2\ U(2\pi)\Phi_1 U(2\pi)^*\
e^{2\pi i (S_{Q+2}-2S_{Q+1}+S_Q)}.
\end{split}
\eeq
Comparing the two different expressions one concludes that $S_q$ has to satisfy
\beq
S_{q+2}-2 S_{q+1}+S_q = 2s , \hspace{1em} \forall q\in\Z.
\eeq
Solving this equation in combination with the conditions $S_0=0, S_{-q}=S_q$ then leads to the desired relation $S_q = s q^2$.
\end{proof}
So we conclude that, although it seems possible at first sight to chose arbitrary spins $S_q$ for every charged sector $\mH_q$, the dependence on the charge $q$ is already determined by the form of the commutation relations \eqref{CommRel}. In addition one gets a kind of spin-statistics relation saying that the real parameter $s$ in $U(2\pi) = e^{2\pi i s Q^2}$ is the same as in the exchange phase $e^{2\pi i s(2N+1)}$ which determines the commutation relations. \\
It therefore suggests itself to define the representation of $\widetilde{U(1)}$ on the full Hilbert space according to
\bneq \label{RotRep}
U(\o) := e^{i s \o Q^2} U_0(\o),
\eneq
where $Q$ is the charge operator, $Q \mH_q = q \mH_q$, and $U_0$ is some ($2\pi$-periodic) representation of $U(1)$, defined naturally for real $\o$ as $U_0(\o) = U_0(\o (\operatorname{mod}2\pi))$. \\
\sloppy{Our field algebra will be constructed on the charged (anti-symmetric) Fock space, i.e. on \mbox{$\mathcal{F}(\mH_1) = \mathcal{F}(\mH_1^+ \oplus \mH_1^-) \simeq \mathcal{F}(\mH_1^+)\otimes\mathcal{F}(\mH_1^-)$} with the usual charge operator and corresponding decomposition into charged sectors.} In this case the above considerations also show that the representation \eqref{RotRep} is \emph{not} the second quantization of some one-particle representation, where the phase would be linear in $q$, and we will have to take care of this peculiarity in our construction of the fields. \\

\section{Construction of the fields}
\subsection{Basic Idea} \label{BasicIdea}
The idea is to first construct auxiliary fields $\hat{\Phi}$ on the circle $S_1 \simeq [0,2\pi)$ by second quantization of specific one-particle operators, which are covariant under the $2\pi$-periodic representation $U_0$ and thus still have the ``wrong'' commutation relations. They are then lifted to the covering space $\widetilde{S_1}$ by using the full representation $U$ of the rotations. More precisely we define for $\o\in\R$,
\bneq
\begin{split}
\Phi_{\o} &:= e^{is\o Q^2}\, \hat{\Phi}_{\o}\, e^{-is\o Q^2} = e^{is\o(2Q-1)} \hat{\Phi}_{\o}, \\
\hat{\Phi}_{\o} &:= U_0(\o) \hat{\Phi}_0 U_0(\o),
\end{split}
\eneq
where $\hat{\Phi}_0$ is a field localized in an interval around the (arbitrary) reference direction $x=0$. By using the representation $U_0$ we get auxiliary fields $\hat{\Phi}_{\o}$ localized in intervals around $x=\hat{\o}$, where $\hat{\o}$ denotes the projection of $\o\in\R$ onto the interval $[0,2\pi)$ (Note that because we chose our standard interval not to be symmetric around $0$, $\hat{\o}$ satisfies $(\widehat{-\o})= 2\pi - \hat{\o}$). \\
One can see that the difference between $\hat{\Phi}_{\o}$ and the final field $\Phi_{\o}$ is just the operator $e^{i\o s(2Q-1)}$ which will lead to an additional phase factor in the commutation relations of the $\Phi$'s.
The important thing is that these auxiliary fields $\hat{\Phi}$ will be defined in such a way that for non-intersecting localization intervals they satisfy commutation relations of the form
\bneq \label{AuxCommRel}
\hat{\Phi}_{\o_1} \hat{\Phi}_{\o_2} = e^{-2is[(\widehat{\o_1-\o_2})-\pi] \pm i\pi}\, \hat{\Phi}_{\o_2} \hat{\Phi}_{\o_1},
\eneq
which still depend explicitly on the relative distance of the respective localization regions. \\
As already stated in section \ref{aim} our basic fields will raise the charge of a vector by one and thus satisfy $\hat{\Phi} Q = (Q-1) \hat{\Phi}$. Using this relations and the definition of the fields $\Phi_{\o}$ one can calculate that their commutation relations turn out to be
\bneq
\Phi_{\o_1} \Phi_{\o_2} = -e^{2is[(\o_1-\o_2)-(\widehat{\o_1-\o_2})+\pi]}\ \Phi_{\o_2} \Phi_{\o_1}.
\eneq
Now remember that $\hat{\o}$ has been defined as $\o\pmod {2\pi}$, which means that there exists an integer $n(\o)\in\Z$ such that
\bneq
\o = \hat{\o} + 2\pi n(\o).
\eneq
We call this $n(\o)$ the \emph{``winding number''} of $\o$ and using its definition we can rewrite the commutation relations according to
\bneq
\Phi_{\o_1} \Phi_{\o_2} = -e^{2\pi is[2n(\o_1-\o_2)+1]}\ \Phi_{\o_2} \Phi_{\o_1}.
\eneq
For intervals $\tI_1, \tI_2 \in \tilde{S_1}$ such that $I_1 \cap I_2 = \emptyset$ the number $n(\o_1-\o_2)$ is constant for all $\o_1\in\tI_1, \o_2\in\tI_2$ allowing us to define a relative winding number of $\tI_1$ w.r.t. $\tI_2$ according to
\bneq \label{IntervalWindingNumber}
N(\tI_1,\tI_2) := n(\o_1-\o_2), \mbox{ for } \o_1\in\tI_1, \o_2\in\tI_2.
\eneq
Hence for fields localized in non-intersecting intervals we get exactly the desired commutation relations \eqref{CommRel}. \\
The interesting question that remains is how we can construct the auxiliary fields $\hat{\Phi}_{\o}$ such that they satisfy the right commutation relations \eqref{AuxCommRel}. In the following we will describe in more detail the explicit construction of these fields by using implementers of certain Bogoliubov transformations.
\\
\subsection{Preliminaries}
Consider the Hilbert space $\mH_1 = L^2(S_1)$ which can be seen as a kind of auxiliary one-particle space. Fourier transformation leads to the equivalence
\bneq
\mH_1 \simeq l^2(\Z) \simeq l^2(\mathds{N}_0)\oplus l^2(\mathds{N}) =: \mH_1^+\oplus \mH_1^-,
\eneq
where w.l.o.g. we have chosen to include the zero-mode into the first summand $\mH_1^+$ ($\mathds{N}_0$ denotes the non-negative integers and $\mathds{N} = \mathds{N}_0\setminus\{0\}$).
On this Hilbert space one has the usual representation of $SO(2)\simeq U(1)$, namely
\bneq
(U_1(\o)\ph)(x) := \ph(x-w),
\eneq
which is diagonal in Fourier space, i.e.
\bneq
(U_1(\o)\tilde{\ph})_n = e^{-in\o}\tilde{\ph}_n.
\eneq
The Fourier modes $\tilde{\ph}_n$ and the inverse transformation are defined according to
\bneq
\tilde{\ph}_n := \frac{1}{\sqrt{2\pi}} \int_0^{2\pi} dx \ph(x) e^{-i x n}, \hspace{1em} \ph(x) = \frac{1}{\sqrt{2\pi}} \sum_{n\in\Z} \tilde{\ph}_n e^{i n x},
\eneq
and to avoid complicating the notation we do not make a notational distinction between the representation $U_1$ in $x$-space and in momentum-space. The components $\ph^\pm$ are then defined acording to
\bneq
\ph^+(x) = \frac{1}{\sqrt{2\pi}} \sum_{n\geq 0} \tilde{\ph}_n e^{inx} , \hspace{1em} \ph^-(x) = \frac{1}{\sqrt{2\pi}} \sum_{n< 0} \tilde{\ph}_n e^{inx}.
\eneq
\sloppy{Over this one-particle space one can now take the anti-symmetrized Fock space \mbox{$\mathcal{F}_a(\mH_1) = \mathcal{F}_a(\mH_1^+\oplus \mH_1^-) \simeq \mathcal{F}_a(\mH_1^+)\otimes\mathcal{F}_a(\mH_1^-)$} and for a unitary operator $U = U_{++}\oplus U_{--}$ on $\mH_1$, which is diagonal w.r.t. the decomposition $\mH_1^+\oplus\mH_1^-$, one defines its second quantization in the usual way},
\bneq
\hG(U) := \Gamma_+(U_{++})\otimes\Gamma_-(\overline{U_{--}}) = \bigoplus_{n=0}^\infty (U_{++})^{\wedge n} \otimes \bigoplus_{n=0}^\infty (\overline{U_{--}})^{\wedge n},
\eneq
where $(\cdot)^{\wedge n}$ stands for the $n$-fold anti-symmetric tensor product. In particular it follows that a constant phase factor $e^{i\gamma}$ is implemented as $\hG(e^{i\gamma}) = e^{i\gamma Q}$, where $Q$ is again the charge operator.\\
As is well known one also has two kinds of creation and annihilation operators $a,a^*$ and $b,b^*$ on this Fock space which can be used to define the free field according to
\bneq
\phi(\ph) := a^*(\ph^+) + b(\overline{\ph^-}) , \hspace{1em} \phi^\dagger(\ph) = b^*(\ph^-)+a(\overline{\ph^+}) , \hspace{2em} \mbox{ for } \ph = \binom{\ph^+}{\ph^-} \in \mH_1
\eneq
W.r.t. this field the second quantization of a unitary diagonal one-particle operator $U$ satisfies
\bneq \label{Implementer}
\hG(U) \phi(\ph) \hG(U)^* = \phi(U\ph)
\eneq
We therefore call the unitary operator $\hG(U)$ the implementer of $U$. For unitary operators $V$ which are \emph{not} diagonal an implementer exists if and only if it satisfies the \emph{Shale-Stinespring criterion}, which demands that the off-diagonal terms $V_{+-}$ and $V_{-+}$ are Hilbert-Schmidt operators. It is clear from equation \eqref{Implementer} that they are only defined up to a constant phase factor whose choice, however, does not alter the commutation relations between implementers. \\
Such implementable fermionic Bogoliubov transformations form a group which we will call $\mathcal{G}_F(\mH_1)$. This group can be decomposed into disconnected components labelled by the Fredholm index of $V_{--}$, i.e.
\bneq
q(V) := \dim \ker (V_{--}) - \dim \ker (V_{--}^*).
\eneq
More importantly it can be shown that the implementer $\hG(V)$ shifts the charge of a vector exactly by the number $q(V)$ (See e.g. \cite{CarRuij87, Ruij77, CarHur85} for a detailed introduction to implementable Bogoliubov transformations in one dimension). Our goal will be to construct fields raising or lowering the charge by one so we are only interested in operators with $q=1$, i.e.
\bneq \label{1DimKer}
\ker(V_{--}) = \{\l e_- \mid \l\in\mathds{C}\}, \hspace{1em} \ker(V_{--}^*) = \emptyset,
\eneq
where $e_-$ is a unit vector in $\mH_1^-$ and we set $e_0 := V e_- \in \mH_1^+$. For such unitaries the explicit form of the implementer turns out to be \cite{CarRuij87}
\bneq \label{ImplementerDef}
\hG(V) = N_V \left[a^*(e_0)\, E_c(Z) + E_c(Z)\, b(\overline{e_-})\right],
\eneq
where $N_V$ is a normalization constant and $E_c(Z)$ is defined according to
\bneq
E_c(Z) := \ :\exp\left[Z_{+-}a^*b^* + (Z_{++}-P_+)a^*a - (Z_{--}-P_-)b b^* - Z_{-+}ba\right]:,
\eneq
which can also be written as
\bneq
E_c(Z) = \exp(Z_{+-}a^*b^*)\Gamma_+(Z_{++})\Gamma(Z_{--}^T)\exp(-Z_{-+}ba).
\eneq
Here $:(...):$ denotes normal ordering and expressions of the form $A a^* b$ are short for
\bneq
A a^* b := \int dpdq\, A(p,q) a^*(p) b(q),
\eneq
where $A(p,q)$ is the integral kernel of the operator $A$ (assuming it exists). The operator $Z$ in the definition \eqref{ImplementerDef} is called the conjugate of $V$  and it depends on $V$ in the following way:
\bneq
\begin{split}
&Z_{++} := -(V_{++}^*)^{-1} ,\ \hspace{1em} Z_{+-} := -(V_{++}^*)^{-1}\ V_{-+}^*, \\
&Z_{-+}:= -V_{--}^{-1}\ V_{-+} , \hspace{1em} Z_{--} := - V_{--}^{-1}.
\end{split}
\eneq
Furthermore, given a bounded self-adjoint operator $A$, the unitary operator $e^{itA}$ for $t\in\R$ is in $\mathcal{G}_F(\mH_1)$ if and only if the off-diagonal elements $A_{+-},A_{-+}$ are Hilbert Schmidt. Such self-adjoint operators form the Lie algebra $\mathfrak{g}_F(\mH_1)$ of $\mathcal{G}_F(\mH_1)$ and due to continuity in $t$ the operators $e^{itA}$ then have vanishing Fredholm index $q(e^{itA})=0$ and therefore their implementers leave the charged sectors invariant. Moreover, there exists a hermitian operator $d\hG(A)$ on the Fock space such that the implementer can be written as $\hG(e^{itA}) = e^{itd\hG(A)}$. Explicit formulas for the implementers in the various cases and their domains of definition can be found e.g. in \cite{Ruij77,Ruij78,CarRuij87} but here we will only need the following properties whose proofs can be found also in \cite{CarRuij87,CarHur85}:
\begin{lemma} \label{ImplementerRelations} ${}$
\begin{enumerate}
    \item Unitaries $V\in\mathcal{G}_F(\mH_1)$ of the form \eqref{1DimKer} with $V_{-+}=0$ create one-particle vectors from the vacuum, more precisely
        \bneq
        \hG(V)\Omega = e_0, \hspace{1em} \mbox{ with } e_0 = V e_-.
        \eneq
    \item For a unitary charge shift $V \in \mathcal{G}_F(\mH_1)$ and self-adjoint operators $A, B \in \mathfrak{g}_F(\mH_1)$, such that $[A,B]=[V,A]=[V,B]=0$ the following commutation relations hold on the Fock space,
        \bneq
        \begin{split}
        \hG(e^{iA}) \hG(V) &= e^{i\langle\hG(V)\Omega,d\hG(A)\hG(V)\Omega\rangle}\ \hG(V) \hG(e^{iA}), \\
        \hG(e^{iA}) \hG(e^{iB}) &= e^{iS(A,B)}\ \hG(e^{iB}) \hG(e^{iA}),
        \end{split}
        \eneq
        where $S(A,B)$ is the so-called ``Schwinger term'' defined according to
        \bneq
        S(A,B) := i\, Tr(A_{-+}B_{+-}-B_{-+}A_{+-}).
        \eneq
    \item The implementers are covariant w.r.t. the second quantization $\hG(U_1(\o)) =: U_0(\o)$ of the (diagonal) one-particle rotations, i.e.
        \bneq
        \hG(U_1(\o))\, \hG(V)\, \hG(U_1(\o))^* = \hG(U_1(\o)V U_1(\o)^*).
        \eneq
\end{enumerate}
\end{lemma}
\emph{Remark: The commonly used definition of the ``Schwinger term'' would be $-i S(A,B)$ in our notation.}\\
An important requirement for being able to compute commutation relations on the Fock space is that the one-particle operators commute. To ensure this we will work only with unitary multiplication operators on $L^2(S_1)$, i.e. operators of the form $(A \ph)(x) = e^{i f(x)} \ph(x) \equiv (e^{if}\ph)(x)$, where we will use the same symbol for the operator on Hilbert space and the function with which it multiplies. For such operators the following lemma characterizes a large class of implementable Bogoliubov transformations \cite{CarHur85}.
\begin{lemma} \label{SmoothFctLemma}
For a smooth real-valued function $\a\in C^\infty(S_1,\R)$ the multiplication operator $(\a\ph)(x) := \a(x)\ph(x)$ has Hilbert Schmidt off-diagonal elements and therefore $e^{it\a}\in\mathcal{G}_F(\mH_1), \forall t\in\R$ and $q(e^{it\a}) = 0$.
\end{lemma}
\begin{proof}
The off-diagonal elements of $\a$ are by definition
\beq
(P_+\a P_-\ph)(x) \equiv (\a_{+-}\ph^-)(x) = \frac{1}{2\pi} \sum_{n\geq 0} \sum_{k<0} \tilde{\a}_{n-k}\tilde{\ph}_k\, e^{inx},
\eeq
so the Hilbert-Schmidt condition reads
\beq
\begin{split}
\operatorname{Tr}[(\a_{+-})^*\a_{+-}] = \operatorname{Tr}[\a_{-+}\a_{+-}] &\propto \sum_{n<0}\sum_{k\geq 0} \tilde{\a}_{n-k} \tilde{\a}_{k-n} = \sum_{n<0}\sum_{k\geq 0} |\tilde{\a}_{n-k}|^2 \\
&= \sum_{k\geq 0} \sum_{n<k} |\tilde{\a}_n|^2 = \sum_{n=1}^\infty n\, |\tilde{\a}_n|^2 < \infty
\end{split}
\eeq
which is fulfilled for smooth functions $\a$.
\end{proof}
Another advantage of using multiplication operators in position space is that a simple expression for the Schwinger term can be computed explicitly \cite{CarHur85}.
\begin{lemma} \label{SchwingerTerm}
Consider self-adjoint operators acting as multiplication with the smooth (real-valued) functions $\a,\beta$ on $L^2(S_1)$. Then the Schwinger term $S(\a,\beta)$ turns out to be
\bneq
S(\a,\beta) = \frac{1}{2\pi} \int_0^{2\pi} dx\ \a(x)\beta '(x)\ = \frac{1}{4\pi} \int_0^{2\pi} dx \left(\a(x)\beta '(x)-\a '(x)\beta(x)\right)
\eneq
\end{lemma}
\begin{proof}
Similar to the proof of lemma \ref{SmoothFctLemma} one calculates
\beq
\begin{split}
i Tr(&\a_{-+}\beta_{+-}-\beta_{-+}\a_{+-}) = \frac{i}{2\pi} \sum_{n<0} \sum_{l\geq 0} \left(\tilde{\a}_{n-l}\,\tilde{\beta}_{l-n} - \tilde{\beta}_{n-l}\,\tilde{\a}_{l-n}\right) \\
&= \frac{i}{2\pi}\sum_{l\geq 0} \sum_{n<-l} \left(\tilde{\a}_n \tilde{\beta}_{-n} - \tilde{\beta}_n \tilde{\a}_{-n}\right) = - \frac{i}{2\pi} \sum_{n=1}^\infty n \left(\tilde{\a}_n \tilde{\beta}_{-n} - \tilde{\beta}_n \tilde{\a}_{-n}\right) \\
&= -\frac{i}{2\pi}\sum_{n\in\Z} n\, \tilde{\a}_n \tilde{\beta}_{-n} = -\frac{1}{2\pi}\sum_{n\in\Z} \widetilde{(\a ')}_n \tilde{\beta}_{-n}.
\end{split}
\eeq
Inverse Fourier transform then leads to the simple expression in position space
\beq
\begin{split}
-\frac{1}{2\pi} \sum_{n\in\Z} \widetilde{(\a ')}_n\tilde{\beta}_{-n} &= -\frac{1}{2\pi}\int_0^{2\pi} dx dy\, \a '(x)\beta(y) \frac{1}{2\pi}\sum_{n\in\Z} e^{-i n(x-y)} \\
&= -\frac{1}{2\pi} \int_0^{2\pi} dx\, \a '(x)\beta(x) = \frac{1}{4\pi}\int_0^{2\pi} dx \left(\a(x)\beta '(x) - \a '(x)\beta(x)\right),
\end{split}
\eeq
where potential boundary terms cancel because of the continuity of $\a$ and $\beta$.
\end{proof}
\vspace{1em}

\subsection{Construction of the Auxiliary Field}
Apart from exponentials of smooth multiplication operators we also need unitaries with non-vanishing Fredholm index. We therefore consider first of all the operator
\bneq
(V\ph)(x) := e^{ix}\ph(x),
\eneq
which is obviously unitary, but the function $[0,2\pi) \ni x \mapsto x$ is \emph{not} smooth on $S_1$ so Lemma \ref{SmoothFctLemma} is not applicable here. However we still have the following result.
\begin{theorem}
The unitary multiplication operator $(V\ph)(x) = e^{ix}\ph(x)$ has the following properties:
\begin{enumerate}
    \item Its off-diagonal elements $V_{\pm\mp}$ are Hilbert-Schmidt.
    \item The diagonal elements satisfy
    \beq
    \ker V_{--} = \{\l e_- |\l\in\R\}, \hspace{1em} \ker V_{--}^* = \ker V_{++} = \emptyset,
    \eeq
    with $e_-(x) = \frac{1}{\sqrt{2\pi}} e^{-ix}$ and $(Ve_-)(x) = e_0(x) = \frac{1}{\sqrt{2\pi}}$.
\end{enumerate}
\end{theorem}
\begin{proof}
Consider the basis $\{e_n \in \mH_1 | n\in\Z\}$ for $\mH_1 = L^2(S_1)$ where $e_n(x) := \frac{1}{\sqrt{2\pi}} e^{i n x}$. Then the operator $V$ acts as a \emph{shift operator} w.r.t. this basis, i.e. it satisfies
\beq
V e_n = e_{n+1}.
\eeq
Now because $\mH_1^+$ is spanned by $\{e_n|n\geq 0\}$ and $\mH_1^-$ by $\{e_n|n<0\}$ it is thus obvious that $e_- \equiv e_{-1}$ spans $\ker V_{--}$ and that $\ker V_{++}=\emptyset$. \\
This also immediately leads to $V_{-+} = 0$ and $Im V_{+-} = \{\l e_0|\l\in\R\}$ so $V_{\pm\mp}$ are evidently Hilbert-Schmidt operators.
\end{proof}
This shows that such a $V$ is a suitable operator to obtain a simple charge shift on the Fock space. One can also define rotated $V$'s according to
\bneq
V_\o = U_1(\o) V U_1(\o)^* = e^{-i\o} V.
\eneq
For the implementers this leads to
\bneq
U_0(\o) \hG(V) U_0(\o)^* = \hG(U_1(\o)V U_1(\o)^*) = \hG(e^{-i\o} V) = \hG(V)e^{-i\o Q}.
\eneq
(Defining it as $ e^{-i\o Q}\hG(V)$ would also be possible, which simply amounts to another choice of phase factor for the implementer.) This simple relation now allows us to compute the commutation relations of different charge shifting operators $\hG(V_{\o_1}), \hG(V_{\o_2})$, which could be a quite involved task in general. Using the charge shifting property of $V$, $V Q = (Q-1) V$, one immediately gets
\bneq
\hG(V_{\o_1}) \hG(V_{\o_2}) = e^{-(\o_1-\o_2)}\, \hG(V_{\o_2}) \hG(V_{\o_1}).
\eneq
These are not yet the right commutation relations we need for our auxiliary field so we consider in addition an operator of the form
\bneq
(e^{i\l\alpha}\ph)(x) := e^{i\l\alpha(x)}\ph(x),
\eneq
where $\alpha \in C^\infty(S_1,\R)$ is a smooth real-valued function on the circle and $\l$ is some real parameter. Lemma \ref{SmoothFctLemma} then tells us that this is an implementable unitary operator allowing us to define the auxiliary field $\hat{\Phi}$ according to
\bneq
\hat{\Phi}_\o := \hG(V_\o) \hG(e^{i\l\alpha_\o}) = U_0(\o)\,\hG(V)\hG(e^{i\l\alpha}) \,U_0(\o)^* ,
\eneq
where $\alpha_\o$ is again defined as $\alpha_\o(x) := \alpha(x-\o)$. \\ \\
The simple form of the operator $V$ and lemma \ref{ImplementerRelations} then allow us to compute the relative commutation relations between $V_{\o '}$ and $e^{i\l\alpha_{\o}}$.
\begin{lemma}
For all smooth real-valued functions $\a$ The commutation relations between $V_{\o '}$ and $e^{i\l\a_\o}$ are independent of $\o ',\o$ and are of the form
\bneq
\hG(e^{i\l\alpha_\o})\hG(V_{\o '}) = e^{i\l\, const.} \hG(V_{\o '}) \hG(e^{i\l\alpha_\o}),
\eneq
where \emph{const.} only depends on the integral over $\a$.
\end{lemma}
\begin{proof}
Since the vector $e_0$, which is created from the vacuum by $\hG(V)$, is invariant under rotations we get
\beq
\langle\hG(V_{\o '})\Omega, d\hG(\alpha_{\o})\hG(V_{\o '})\Omega\rangle = \langle e_0,\alpha_\o e_0\rangle = \langle e_0,\alpha\, e_0\rangle = \mbox{const.}
\eeq
\end{proof}
These constant phase factors will therefore cancel out in the total commutation relations between the $\hat{\Phi}$'s which then turn out to be
\bneq
\hat{\Phi}_{\o_1} \hat{\Phi}_{\o_2} = e^{-i(\o_1-\o_2)}e^{i\l^2 S(\alpha_{\o_1},\alpha_{\o_2})}\, \hat{\Phi}_{\o_2} \hat{\Phi}_{\o_1}.
\eneq
Comparing this with the relations in section \ref{BasicIdea} we are now faced with the following problem: \vspace{1em} \\
\emph{Find a smooth real-valued function $\alpha$ on the circle and a parameter $\l\in\R$  for which the Schwinger term $S(\alpha_{\o_1},\alpha_{\o_2})$ satisfies
\bneq \label{SchwingerTermRelation}
\l^2 S(\alpha_{\o_1},\alpha_{\o_2}) -(\widehat{\o_1-\o_2}) = -2s[(\widehat{\o_1-\o_2})-\pi] \pm \pi,
\eneq
for suitable $(\widehat{\o_1-\o_2})$.}

\vspace{1em}
For this purpose consider first the function $x\mapsto \lambda \hat{x}$ on $\R$ for an abitrary real parameter $\l$. Similar to the shift operator $V$, multiplication with $e^{i\lambda \hat{x}}$ would lead to commutation relations with exchange phase of the form $e^{-i\lambda(\o_1-\o_2)}$. But the function $e^{i\l\hat{x}}$  is in general not continuous on the circle and therefore it does not lead to an implementable transformation. This problem can be solved by smearing it with a function $\chi_\e \in C_0^\infty(\R)$ with the following properties:
\begin{wrapfigure}[4]{r}[0pt]{0.3\textwidth}
     \vspace{-1em}\includegraphics[width=0.3\textwidth]{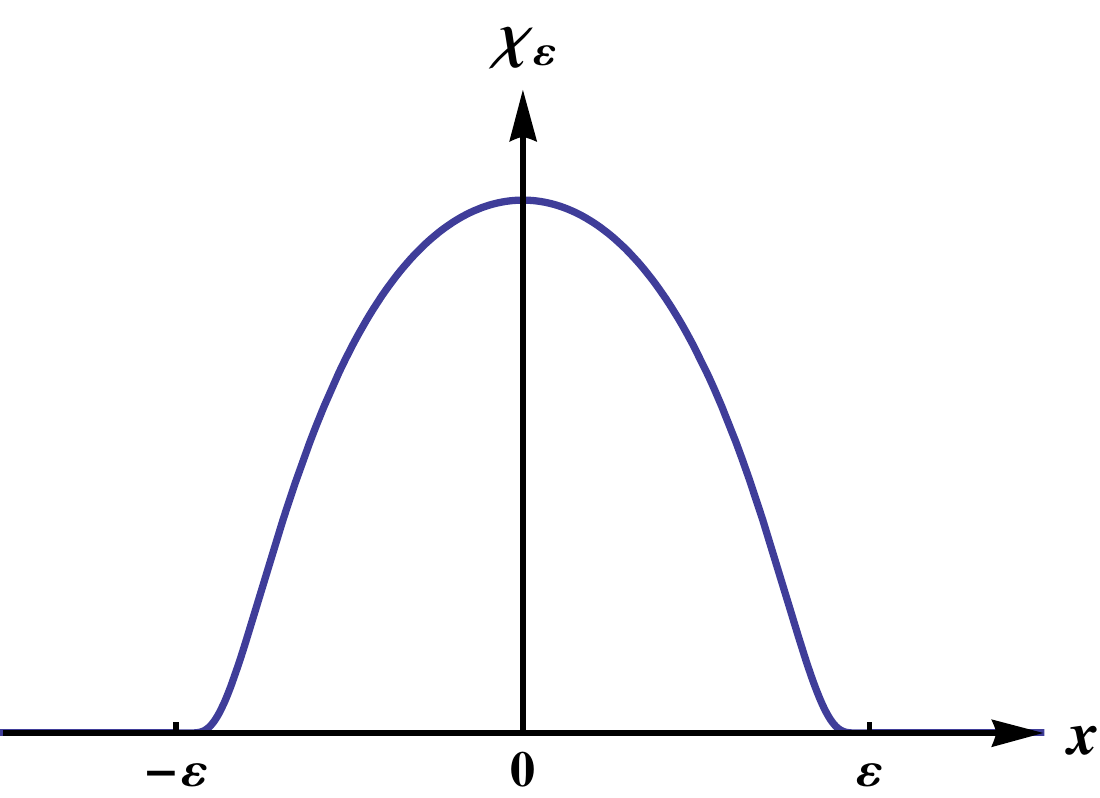}
\end{wrapfigure}
\begin{itemize}
    \item $\operatorname{supp}\chi_\e = [-\e,\e]$ with $0<\e<\frac{\pi}{2}$
    \item $\chi_\e(-x) = \chi_\e(x) = \overline{\chi_\e(x)}$
    \item $\int dx \chi_\e(x) = 1$
\end{itemize}
\vspace{2em}
Considering the $2\pi$-periodic function $\hat{x} \equiv x\pmod {2\pi}$ on $\R$ one can use such a $\chi_\e$ to define a smooth function according to \\
\vspace{-2em}\begin{wrapfigure}[8]{r}[0pt]{0.3\textwidth}
     \includegraphics[width=0.3\textwidth]{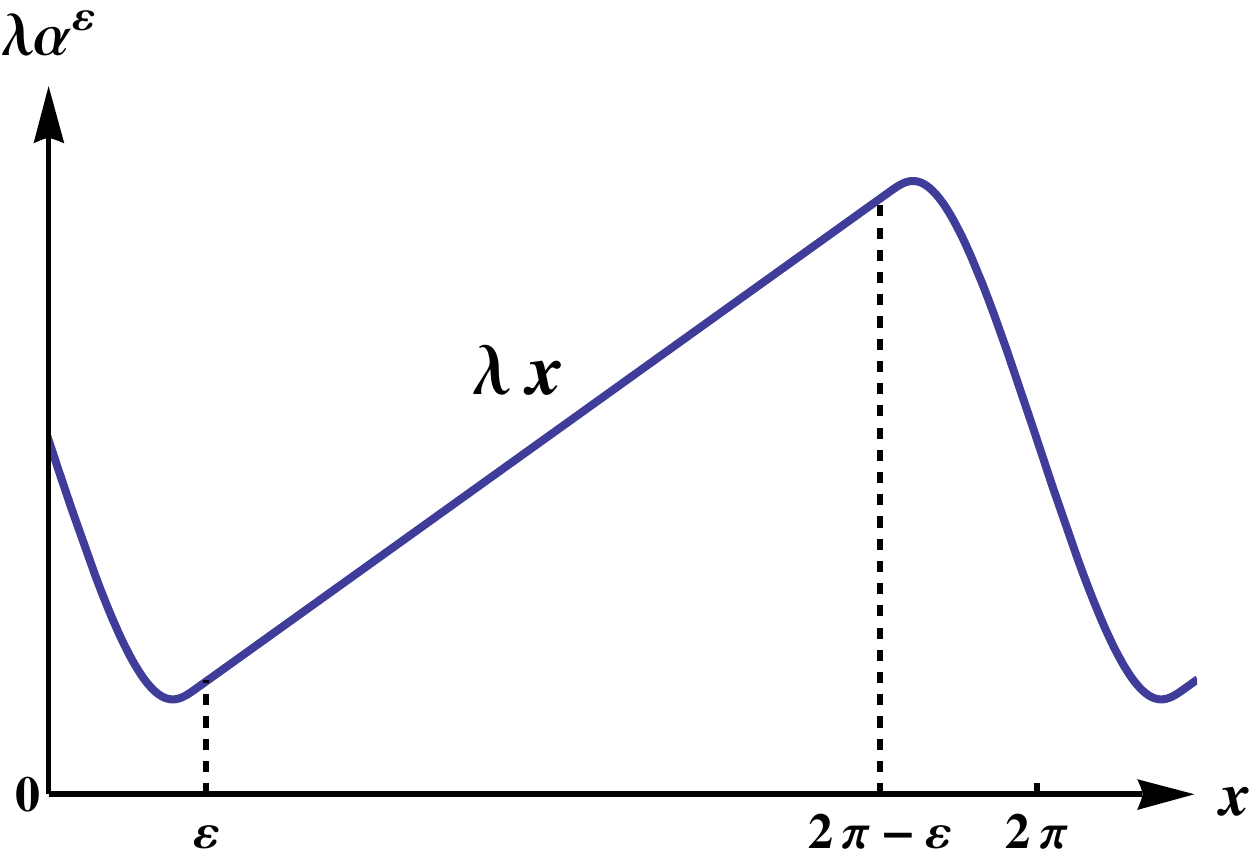}
\end{wrapfigure}
\vspace{1em}
\bneq \label{DefAlpha}
\a^\e(x) := \int_\R dy\, \widehat{(x-y)}\chi_\e(y),
\eneq \vspace{0em}

which is still $2\pi$-periodic and thus defines a smooth function on the circle. An important feature of this function is that it is still linear in the interval $(\e,2\pi-\e)$, because for our choice of a symmetric $\chi_\e$ there holds
\beq
\int dy\, (x-y)\chi_e(y) = x,
\eeq
which means that most of the linear part of $\hat{x}$ is unchanged and only the jump discontinuities get smeared. \\
For every real parameter $\l\in\R$ this $\a^\e$ then defines an implementable unitary multiplication operator according to
\bneq
(e^{i\l\a^\e}\ph)(x) = e^{i\l\a^\e(x)}\ph(x).
\eneq
It turns out that the Schwinger term for such operators has exactly the right form, namely we have the following proposition.
\begin{theorem}
For functions $\alpha^{\e_1}, \alpha^{\e_2}$ with $0 < \e_i < \frac{\pi}{2}, i= 1,2$ as defined in equation \eqref{DefAlpha} and $\o_1,\o_2 \in \R$ such that $\e_1+\e_2 < (\widehat{\o_1-\o_2}) < 2\pi - \e_1 - \e_2$ the Schwinger term satisfies
\bneq
S(\alpha_{\o_1}^{\e_1},\alpha_{\o_2}^{\e_2}) = (\widehat{\o_1-\o_2})-\pi.
\eneq
\end{theorem}
\begin{proof}
Evidently the Schwinger term is invariant under simultaneous rotation of $\a_{\o_1}^{\e_1}$ and $\a_{\o_2}^{\e_2}$ so for simplicity we write $\o \equiv \o_1-\o_2$ and calculate
\beq
S(\a_\o^{\e_1},\a^{\e_2}) = \frac{1}{2\pi} \int_0^{2\pi} dx\, \a_\o^{\e_1}(x) (\a^{\e_2})'(x).
\eeq
For this we first need $(\a^\e)'$ which turns out to be
\beq
\begin{split}
(\a^\e)'(x) &= \int dy\, \frac{d}{dx}(\widehat{x-y}) \chi_\e(y) \\
&= \int dy\left(1-2\pi\sum_{k\in\Z}\delta(x-y-2\pi k)\right)\chi_\e(y) \\
&= 1-2\pi\sum_{k\in\Z}\chi_\e(x-2\pi k).
\end{split}
\eeq
Note that, because of the support properties of $\chi_\e$, for every $x\in\R$ only one term in the sum $\sum_{k\in\Z}\chi_\e(x-2\pi k)$ is nonzero. Hence one gets
\beq
\begin{split}
S(\alpha_{\o_1}^{\e_1},\alpha_{\o_2}^{\e_2}) &= \int_0^{2\pi} dx\, \a_\o^{\e_1}(x)\big(\frac{1}{2\pi} - \chi_{\e_2}(x)-\chi_{\e_2}(x-2\pi)\big) \\
&= \pi - \int_{-\e_2}^{\e_2} dx \int_{-\e_1}^{\e_1} dy\, (\widehat{x-y-\o})\chi_{\e_1}(y)\chi_{\e_2}(x),
\end{split}
\eeq
by using the periodicity of $\alpha^{\e_1}$ and the fact that $\int dx\, \alpha_\o^{\e_1}(x) = 2\pi^2$ is independent of $\o$ and $\chi_{\e_1}$. To compute the remaining term we need to calculate integrals of the form
\beq
\int_{-\e}^{\e}dx\, (\widehat{c-x})\chi_{\e}(x), \hspace{1em} \mbox{ for } c \in \R \mbox{ and } \e < \hat{c} < 2\pi - \e.
\eeq
For this purpose remember that we can write $(\widehat{c-x}) = (c-x) - 2\pi n(c-x)$ and that $n(c-x) = n(c)$ for $\e<\hat{c}<(2\pi -\e)$ and $-\e < x < \e$. Inserting this into the integral we get
\beq
\int_{-e}^{\e}dx\, (\widehat{c-x})\chi_\e(x) = \int_{-\e}^{\e} dx\, \big((c-x)-2\pi n(c)\big)\chi_\e(x) = c-2\pi n(c) = \hat{c}.
\eeq
Using this in the expression for $S(\alpha_{\o_1}^{\e_1},\alpha_{\o_2}^{\e_2})$ one immediately arrives at
\beq
S(\alpha_{\o_1}^{\e_1},\alpha_{\o_2}^{\e_2}) = \pi - (\widehat{-\o}) = \hat{\o}-\pi.
\eeq
\end{proof}
With this proposition equation \eqref{SchwingerTermRelation} then leads to the constraint
\bneq
\begin{split}
\l^2 (\hat{\o} - \pi) - \hat{\o} \overset{!}{=} -2s(\hat{\o}-\pi) \pm \pi
\end{split}
\eneq
for the parameters $\l$ and $s$. In order to get a solution for $\l$ we have to choose the minus sign on the right side and for a fixed spin $s$ this equation then restricts the parameters to
\bneq
\l^2 = 1-2s.
\eneq
Together with the requirement that the operator $e^{i\l\a^\e}$ has to be unitary this also shows that the ``spin'' parameter $s$ has to satisfy $s\in (-\infty,\frac{1}{2})$. \\ \\
We can now summarize the construction in the following way: For every $\o\in\R$, $\e>0$ and symmetric real-valued smearing function $\chi_\e$ with $\operatorname{supp}\chi_\e \in (-\e,\e)$ one can construct field operators
\bneq
\Phi_\o[\chi_\e] = U(\o)\hG(V)\hG(e^{i\l\a^\e})U(\o)^* = e^{is\o(2Q-1)}\hG(V_\o)\hG(e^{i\l\a_\o^\e}),
\eneq
which raise the charge by one and are localized in the interval
\bneq
\widetilde{I}^\e_\o := \{x\in\widetilde{S_1}|\, \o-\e<x<\o+\e\} \subset \widetilde{S_1}
\eneq
with width $2\e$ centered around the point $\o \in \R$. Together with the adjoint field $\Phi_\e[\chi_\e]^*$, lowering the charge by one, these operators then generate a net of algebras on the space $\widetilde{S_1}$ with anyonic statistics.
\begin{lemma}
Consider field operators $\Phi_i = \Phi_{\o_i}[\chi_{\e_i}], i=1,2$, which are localized in intervals $\tI_1$ and $\tI_2$ respectively. If the intervals are non-intersecting, i.e. $I_1\cap I_2=\emptyset$, the fields satisfy anyonic commutation relations
\bneq \begin{split}
\Phi_1\, \Phi_2 &= -e^{2\pi i s(2 N(\tI_1,\tI_2)+1)}\, \Phi_2\, \Phi_1, \\
\Phi_1\, \Phi_2^* &= -e^{-2\pi i s(2 N(\tI_1,\tI_2)+1)}\, \Phi_2^*\, \Phi_1,
\end{split} \eneq
where $N(\tI_1,\tI_2)$ is the relative winding number of $\tI_1$ w.r.t $\tI_2$, defined in \eqref{IntervalWindingNumber}.
\end{lemma}
In addition the fields are covariant with respect to the unitary representation \eqref{RotRep} of the universal covering of the rotation $\widetilde{U(1)}$ with real-valued spin $s\subset (-\infty,\frac{1}{2})$. Taking the polynomial algebra over such localized fields then defines the local algebras $\mathcal{F}(\tI)$ and we have therefore constructed a local, covariant quantum field net for Anyons on the circle satisfying all the requirements of section \ref{aim}.

\vspace{1em}
\subsection{Special Cases}
\textbf{s=1/2:} \\
For the maximal value $s=\frac{1}{2}$ the relation $\l^2 = 1-2s$ leads to $\l = 0$, which means the operator $e^{i\l\a^\e}$ is the identity operator in this case. We are therefore left with only the shift operator $V$ and the field is
\bneq
\Phi_\o = e^{i\frac{\o}{2}(2Q-1)}\hG(V)e^{-i\o Q} = e^{i\frac{\o}{2}}\hG(V),
\eneq
which is just $\hG(V)$ times a constant phase factor which could also be omitted. This shows that the fields for different $\o$'s commute independently of $\o$ which is in accordance with the fact that the exchange phase turns out to be
\beq
-e^{2\pi i s(2n+1)} = -e^{\pi i (2n+1)} = 1.
\eeq
So we see that for $s=\frac{1}{2}$ we get a \emph{bosonic field} ``$\Phi \sim \hG(V) \sim \hG(e^{ix})$'' and because of its simple form it creates one-particle vectors from the vacuum, namely
\beq
\Phi_{\o} \Omega = e^{i\frac{\o}{2}}\hG(V)\Omega = e^{i\frac{\o}{2}} e_0.
\eeq
Moreover, taking into account that $Z_{\pm\mp}=0$ for $V_{-+}=0$, the explicit form of $\hG(V)$ (with an appropriate choice of phase) turns out to be
\bneq
\hG(V) = a^*(e_0)\Gamma_+(-V_{++})\Gamma_-(-\overline{V_{--}}) + \Gamma_+(-V_{++})\Gamma_-(-\overline{V_{--}}) b(\overline{e_-}),
\eneq
which looks like a free field modified by the unitary operator $\Gamma_+(-V_{++})\Gamma_-(-\overline{V_{--}})$. \\

\begin{wrapfigure}{r}[0pt]{0.4\textwidth}
     \includegraphics[width=0.4\textwidth]{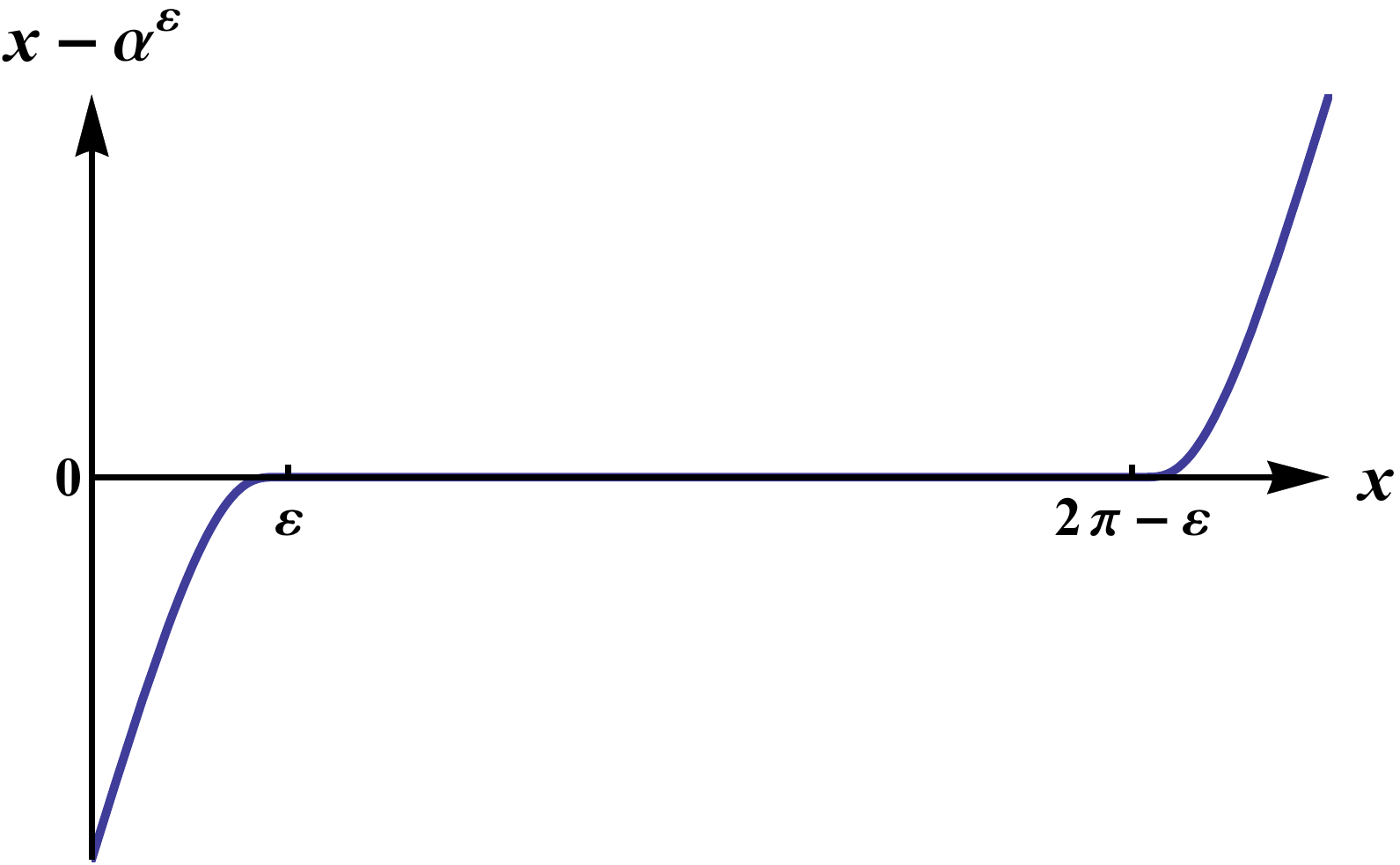}
\end{wrapfigure}
\textbf{s=0:} \\
The case $s=0$ leads via $\l^2 = 1-2s$ to $\l = \pm 1$ and we will only consider the case $\l=-1$. As was already shown the function $\a^\e$ is linear in the interval $(\e,2\pi-\e)$ so for $\l=-1$ the function $x\mapsto\hat{x}+\l\a^\e(x) = \hat{x}-\a^\e(x)$ vanishes for $\e<x<2\pi-\e$. Therefore $e^{i x}e^{-i\a^\e(x)}$ is simply the identity apart from a small interval of length $2\e$ where its phase changes by the value $2\pi$. \\
A similar so-called ``blip function'' was used in \cite{Seg81} and \cite{CarLang98,CarLang10} to approximate a step-function on the circle. There it was shown that after taking an appropriate limit the implementer of such a function converges to the free fermi field on the circle. Because of
\beq
-e^{2\pi i s(2n+1)} = -1, \mbox{ for } s=0,
\eeq
we will also get an anti-commuting field in our case, but it is unclear if the field converges in some sense to a free fermi field if the smearing function $\chi_\e$ tends to a delta function. \\ \\
\emph{Remark:} As we have seen the construction leads to a commuting field for $s=\frac{1}{2}$ and an anti-commuting field for $s=0$. Howerver, we are working on a one-dimensional space --- the circle $S_1$ or its universal covering $\widetilde{S_1}$ respectively --- so although we called the parameter $s$ the ``spin'' it really just labels a representation of the translations on $S_1$ or $\widetilde{S_1}$. So one would not expect the usual kind of spin-statistics theorem to hold for $s$ in our case.
\vspace{1em}

\section{Non-relativistic ``String-local'' Fields on $\R^2$}
It would now be tempting to try the same construction in higher dimensions, e.g. on the Hilbert space $L^2(\R^2)$, to construct a string-local quantum field in two (space-)dimensions, covariant under the Euclidean group $E(2)$ or its universal cover $\widetilde{E(2)}$ respectively. However, such direct attempts are facing serious difficulties concerning either covariance or the Hilbert-Schmidt condition of the occurring operators, because the method of implementing multiplication operators as Bogoliubov transformations is basically restricted to one dimension. We will therefore try to circumvent these problems by considering simply a tensor product of a local field on $\R^2$ with the previously constructed circle-fields. \\ \\
For this purpose consider the following field on the anti-symmetric Fock space $\mathcal{F}_a(L^2(\R^2))$,
\bneq
\Psi(f) := c^*(f) + c(\overline{f}) , \hspace{1em} \Psi(\overline{f})^* = \Psi(f),
\eneq
for $f \in L^2(\R^2)$, where $c$ and $c^*$ are the usual annihilation and creation operators on $\mathcal{F}_a(L^2(\R^2))$. From the anti-commutation relations of $c$ and $c^*$,
\beq
 \{c(f),c(g)\} = 0, \hspace{1em} \{c(f),c^*(g)\} = \langle f,g\rangle,
\eeq
it follows that this field satisfies
\beq
\Psi(f) \Psi(g) = - \Psi(g) \Psi(f),
\eeq
if the testfunctions $f$ and $g$ are such that $\operatorname{supp} f \cap \operatorname{supp} g = \emptyset$. In addition $\Psi$ is covariant with respect to to the second quantization of the pullback representation $\mathcal{U}$ of $E(2)$ defined on $L^2(\R^2)$ according to
\bneq
(\mathcal{U}(\vec{a},\o)f)(\vec{x}) := f(R(-\o)(\vec{x}-\vec{a})),
\eneq
where $R(\o)$ is the usual rotation matrix acting on vectors in $\R^2$. Now take as Hilbert space the tensor product
\bneq
\mH = \mathcal{F}_a(L^2(\R^2))\otimes\mathcal{F}_a(L^2(S_1)),
\eneq
with a representation of $\widetilde{E(2)}$ of the form
\bneq \label{TensorRepresentation}
\hG(\mathcal{U}(\vec{a},\o))\otimes e^{is\o Q^2}\hG(U_1(\o)).
\eneq
For every $f\in L^2(\R^2)$, admissible $\chi_\e$ and $\o\in\R$ one can then define on this space the new fields
\bneq \label{DefCombinedFields}
F_\o[f,\chi_\e] := \Psi(f)\otimes \Phi_\o[\chi_\e].
\eneq
They of course inherit the anyonic commutation relations and are covariant under the representation \eqref{TensorRepresentation}, where the translations only act on the first tensor factor (shifting the support of the test function) and the rotations act on both. The motivation behind this definition is that these field operators can be interpreted as being localized in \emph{conelike regions} on the two-dimensional plane. More specifically consider the following subset of $\R^2$,
\bneq \label{DefCones}
C[f,I_\o^\e] := \operatorname{supp}f + \R_+ \bigcup_{\mu\in I^\e_\o} \vec{n}_\mu \subset \R^2, \footnote{Expressions of the form $\R_+ \mathcal{O}$ for $\mathcal{O}\subset \R^2, \vec{0}\notin \mathcal{O}$ are shorthand for $\{\nu\, \mathcal{O}|\ \nu \in \R, \nu >0\} \subset \R^2$.}
\eneq
where $\vec{n}_\mu \in \R^2$ is a unit vector in the direction $\mu$, i.e.
\beq
\vec{n}_\mu := R(-\mu)\vec{n}_0,
\eeq
with a standard unit-vector $\vec{n}_0$, e.g. $\vec{n}_0 = \binom{0}{1}$. This defines a cone-shaped region which ``starts'' in $\operatorname{supp}f$ and extends to infinity in the set of directions given by $I_\o^\e$. Since a translation only acts on the test function $f$ it shifts the whole cone, whereas a rotation also changes the interval $I$ corresponding to the asymptotic directions of the cone $C[f,I]$. The interpretation of such cones as localization regions for the fields \eqref{DefCombinedFields} is then possible because of the following commutation relations, which result from the anti-commutativity of the local field $\Psi$ and the anyonic commutation relations of the circle field $\Phi$.
\begin{theorem}
For compactly localized test functions $f_1, f_2$ and intervals $\tI_1 = \tI_{\o_1}^{\e_1}$ and $\tI_2 = \tI_{\o_2}^{\e_2}$ such that the corresponding cones do not intersect, i.e.
\beq
C[f_1,I_1] \cap C[f_2,I_2] = \emptyset,
\eeq
the fields $F_1 = F_{\o_1}[f_1,\chi_{\e_1}]$ and $F_2 = F_{\o_2}[f_2,\chi_{\e_2}]$ satisfy
\bneq \begin{split}
F_1\, F_2 &= e^{2\pi i s (2 N(\tI_1,\tI_2)+1)}\ F_2\, F_1, \\
F_1\, F_2^* &= e^{-2\pi i s (2 N(\tI_1,\tI_2)+1)}\ F_2^*\, F_1.
\end{split} \eneq
\end{theorem}
\begin{proof}
From the definition \eqref{DefCones} it is clear that if the cones $C[f_1,I_1]$ and $C[f_2,I_2]$ are disjoint then also the testfunction supports and intervals have to satisfy\footnote{Note that the converse is not true, namely there can be non-overlapping test functions and intervals such that the corresponding cones actually have a (finite) overlap.}
\beq
\operatorname{supp}f_1 \cap \operatorname{supp}f_2 = \emptyset, \mbox{ and }
I_1\cap I_2 = \emptyset.
\eeq
But under this conditions we get that
\beq
\begin{split}
\Psi(f_1) \Psi(f_2) &= (-1) \Psi(f_2) \Psi(f_1), \mbox{ and} \\
\Phi_{\o_1}[\chi_{\e_1}]\, \Phi_{\o_2}[\chi_{\e_2}] &= (-1) e^{2\pi i s (2 N(\tI_1,\tI_2)+1)}\ \Phi_{\o_2}[\chi_{\e_2}]\, \Phi_{\o_1}[\chi_{\e_1}].
\end{split}
\eeq
Considering the simple tensor product structure of the fields $F_\o[f,\chi_\e]$ then leads to the asserted commutation relations.
\end{proof}

\emph{Remark:} In addition to the conelike localization regions $C[f,I_\o^\e]$ \eqref{DefCones} (which are obviously invariant under $\o \mapsto \o + 2\pi$), the fields also depend on the winding number $n(\o)$ of $\o$, which determines the commutation relations. To account for this fact the fields can be interpreted to be localized in \emph{``generalized cones''} (or ``paths of cones'', see e.g. \cite{Mund03}) which are defined in the following way. A usual cone $C$ in two dimensions is determined by a point $\vec x\in \R^2$ (its apex) and an interval $I$ on the circle, specifying the asymptotic directions contained in $C$. Hence one can denote a cone by the pair $C = (\vec x, I)$, where the width of $I$ determining the opening angle of $C$ should be smaller than $\pi$. \par
If we now allow for generalized intervals $\tI$ on the universal covering $\widetilde{S_1}$ of the circle one can define generalized cones $\tilde{C}$ as pairs
\bneq
\tilde{C} = (\vec x,\tI), \hspace{1em} \vec x\in \R^2, \tI \subset \widetilde{S_1}.
\eneq
After smearing in $\vec x$ with a testfunction $f$ we can obviously also define ``smeared'' generalized cones as $\tilde{C} = (f,\tI) \equiv (\operatorname{supp} f,\tI)$. \\
\begin{figure}[H]
\centering
\begin{tikzpicture}[scale=0.8]
\begin{scope}
\fill[black!10] (2,2) -- (3,1) -- (6,2) -- (4.4,4.4) -- cycle;
\draw[dashed] (0,0) -- (2,2);
\draw[dashed] (0,0) -- (3,1);
\draw[fill=black!20] plot [smooth cycle] coordinates {(2.1,1.3) (2,2) (3.2,1.8) (3,1)};
\draw node at (2.6,1.6) {$f_2$};
\draw node at (4.2,2.5) {$\tilde{\mathcal{C}}(f_2,\tilde{I_2})$};
\draw[thick] (2,2) -- (4.4,4.4);
\draw[thick] (3,1) -- (6,2);
\draw[dashed] (0,0) circle(2);
\draw[thick] (1.414,1.414) arc(45:18.4:2);
\draw[thick] (1.334,1.334) -- (1.494,1.494);
\draw[thick] (1.78,0.59) -- (2.025,0.67);
\draw node at (1.4,0.9) {$\tilde{I_2}$};
\draw[dashed] (0,0) -- (0,-2.1);
\draw node at (-0.3,-1.8) {\small$\o_0$};
\draw[->] (0,-1.2) arc(-90:30:1.2);
\end{scope}
\begin{scope}[shift={(-5,1)},rotate=80,scale=0.9]
\fill[black!10] (2,2) -- (3,1) -- (6,2) -- (4.4,4.4) -- cycle;
\draw[dashed] (0,0) -- (2,2);
\draw[dashed] (0,0) -- (3,1);
\draw[fill=black!20] plot [smooth cycle] coordinates {(2.1,1.3) (2,2) (3.2,1.8) (3,1)};
\draw node at (2.6,1.6) {$f_1$};
\draw node at (4.2,2.5) {$\tilde{\mathcal{C}}(f_1,\tilde{I_1})$};
\draw[thick] (2,2) -- (4.4,4.4);
\draw[thick] (3,1) -- (6,2);
\draw[dashed] (0,0) circle(2);
\draw[thick] (1.414,1.414) arc(45:18.4:2);
\draw[thick] (1.334,1.334) -- (1.494,1.494);
\draw[thick] (1.78,0.59) -- (2.025,0.67);
\draw node at (1.4,0.9) {$\tilde{I_1}$};
\draw[dashed] (0,0) -- (-2,-0.3);
\draw node at (-1.7,0) {\small$\o_0$};
\draw[->] (-1.1,-0.16) arc(-180:40:1.2);
\end{scope}
\end{tikzpicture}
\caption{Non-intersecting generalized cones $\tilde{\mathcal{C}}(f_1,\tilde{I_1})$ and $\tilde{\mathcal{C}}(f_2,\tilde{I_2})$}
\end{figure}
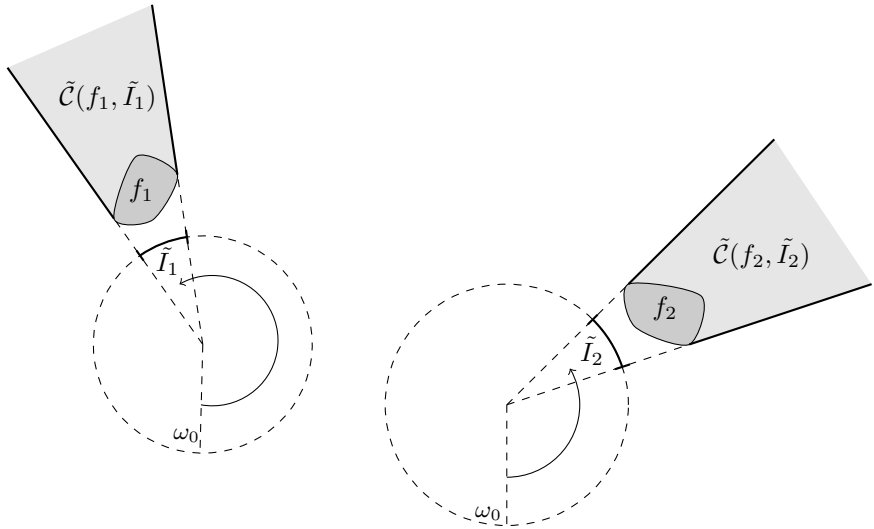 \vspace{1em}

We can therefore construct again, for every generalized cone $\tilde{C}$, the polynomial algebra $\mathcal{F}(\tilde{C})$ of fields localized in $\tilde{C}$. This net $\tilde{C} \mapsto \mathcal{F}(\tilde{C})$ then satisfies the same requirements for an Anyonic field net as defined in section \ref{aim}, with the obvious generalizations that we replace intervals $\tilde{I}$ with cones $\tilde{C}$ and covariance holds w.r.t. a representation of $\widetilde{E(2)}$ instead of just $\widetilde{U(1)}$. \\ \vspace{1em}

This construction now has an evident generalization to an arbitrary (interacting) field algebra on three dimensional Minkowski space in the following way. First note that one can define a generalized cone also on $\R^{1+2}$ as a pair $(\vec{x},\tI)$, where $\vec{x}\in\R^{1+2}$ is the apex of the cone. Consider then a local and covariant (Bosonic or Fermionic) field net $\mathcal{O} \mapsto \mathcal{F}(\mathcal{O})$ on $\R^{1+2}$, indexed by double cones $\mathcal{O} \subset \R^{1+2}$. The field algebra $\mathcal{F}(\tilde{C})$ for a an arbitrary generalized cone $\tilde{C} \equiv \tilde{C}(\vec{x},\tI)$ can then be defined as
\beq \mathcal{F}(\tilde{C}) = \bigcup_{\mathcal{O}\subset C} \mathcal{F}(\mathcal{O})\otimes\mathcal{F}(\tI), \eeq
where $\mathcal{F}(\tI)$ is just the algebra of fields on the (covering of the) circle localized in $\tI$. Note that in this definition we had to take the union over all double cones contained in $C$ in order to ensure isotony of the resulting net. This leads again to a $\widetilde{E(2)}$-covariant and twisted-local field net, localized in generalized cones $\tilde{C}$  with winding-number dependent commutation relations. \par
If the original net was interacting then also the new composite net has a nontrivial scattering matrix. However, in the construction of the scattering states one has to take into account the extended localization regions of the fields and the new S-matrix should also depend on the relative winding number of the occurring fields (see e.g. \cite{FreGabRue96, FroeMar91, BrosMund12}).


\vspace{1em}

\section{Summary and Outlook}
We have seen that it is possible to construct for every interval $\tilde{I}$ on the universal covering of the circle $\widetilde{S_1}$ a field algebra $\mathcal{F}(\tilde{I})$, such that the resulting net $\tilde{I} \mapsto \mathcal{F}(\tilde{I})$ is covariant under a representation of $\widetilde{U(1)}$ for arbitrary real-valued spin. For non-intersecting intervals the corresponding field algebras satisfy twisted commutation relations which only depend on the relative winding number of the intervals. This is achieved by considering unitary implementers of certain one-particle multiplication operators where the occurring Schwinger term is used to change the commutation relations of the operators in the desired way. In contrast to previous similar constructions this leads to a \emph{local} quantum field theory for Anyons on the circle. Moreover, everything is defined explicitly and non-perturbatively on the well known anti-symmetric Fock space and no taking of limits (or thereby leaving the Fock space) is needed. \\
Taking tensor products of these ``circle Anyons'' with any local covariant quantum field theory one can obtain a cone-local theory with Anyonic commutation relations. However, potential boost-covariance of the original field theory is lost so this only leads to a non-relativistic theory (covariant under $\widetilde{E(2)}$ plus potential time translations). A construction of a cone-local Anyonic field net, which is covariant under the full Poincaré group in $1+2$ dimensions would be desirable but unfortunately this cannot be achieved by the methods described here. The first problem is that the ``trick'' using an auxiliary ($2\pi$-periodic) field which then gets lifted to a covering space only works in this simple way for the pure rotation group. For the covering of the full Lorentz group $\tilde{\mathcal{L}}_+^\uparrow$ one would get a representation on Fock space of the form $U(\tilde{\Lambda}) \sim e^{isQ\Omega(\tilde{\Lambda})} U_0(\Lambda)$ where $\Omega(\tilde{\Lambda})$ is now an operator on the Hilbert space instead of a mere constant $\omega$ as for the group $\widetilde{U(1)}$ (see e.g. \cite{Mund03} for a possible representation of $\tilde{\mathcal{L}}_+^\uparrow$ on the mass shell). Another problem is that the method of considering multiplication operators on the one-particle Hilbert space as implementable Bogoliubov transformations leads to problems concerning the Hilbert-Schmidt property for theories in more than one dimension (see e.g. \cite{Lang94}). \\
Nevertheless, the above construction provides a simple example of a (cone-)local covariant -- and possibly interacting -- quantum field net which exhibits the main features expected in a full $2+1$ dimensional theory of Anyons, namely the relation between non-trivial behavior under $2\pi$-rotations and the dependence on some winding number of the cone-shaped localization regions.

\clearpage
\subsection*{Acknowledgments}
I would like to thank the Vienna Mathematical Physics group, especially Sabina Alazzawi, Christian Köhler and Jan Schlemmer for numerous supporting discussions. I am furthermore indebted to my advisor Jakob Yngvason for providing constant support. I also want to thank Alan Carey, Harald Grosse and Edwin Langmann for helpful comments on my work. \\
This work was partly supported by the FWF-project P22929-N16.

\end{document}